\documentclass[12pt]{article}

\usepackage{amsmath}
\usepackage{amssymb}
\usepackage{amsfonts}
\usepackage{latexsym}
\usepackage{color}

\catcode `\@=11 \@addtoreset{equation}{section}

\catcode `\@=12

\newcommand{\be}{\begin{equation}}
\newcommand{\en}{\end{equation}}
\newcommand{\bea}{\begin{eqnarray}}
\newcommand{\ena}{\end{eqnarray}}
\newcommand{\beano}{\begin{eqnarray*}}
\newcommand{\enano}{\end{eqnarray*}}
\newcommand{\bee}{\begin{enumerate}}
\newcommand{\ene}{\end{enumerate}}

\newcommand{\mc}{\mathcal}

\newcommand{\D}{{\mc D}}
\newcommand{\V}{{\mc V}}

\newcommand{\Sc}{{\cal S}}
\newcommand{\E}{{\cal E}}
\newcommand{\F}{{\cal F}}
\newcommand{\G}{{\cal G}}

\newcommand{\Lc}{{\cal L}}

\newcommand{\1}{1 \!\! 1}

\newcommand{\Hil}{\mc H}

\newtheorem{thm}{Theorem}

\newtheorem{lemma}[thm]{Lemma}
\newtheorem{prop}[thm]{Proposition}
\newtheorem{defn}[thm]{Definition}

\newenvironment{proof}{\noindent {\bf Proof --}}{\hfill$\square$ \vspace{3mm}\endtrivlist}

\catcode `\@=11 \@addtoreset{equation}{section}
\catcode `\@=12

\begin{document}

\thispagestyle{empty}

\vspace*{2cm}

\begin{center}
{\Large \bf Pseudo-bosons and bi-coherent states out of $\Lc^2(\mathbb{R})$}   \vspace{2cm}\\

{\large F. Bagarello}\\
  Dipartimento di Ingegneria,
Universit\`a di Palermo,\\ I-90128  Palermo, Italy\\
and I.N.F.N., Sezione di Napoli\\
e-mail: fabio.bagarello@unipa.it\\
home page: www1.unipa.it/fabio.bagarello

\end{center}

\vspace*{2cm}

\begin{abstract}
\noindent In this paper we continue our analysis on deformed canonical commutation relations and on their related pseudo-bosons and bi-coherent states. In particular,  we show how to extend the original approach outside the Hilbert space $\Lc^2(\mathbb{R})$, leaving untouched the possibility of defining eigenstates of certain number-like operators, manifestly non self-adjoint, but opening to the possibility that these states are not square-integrable. We also extend this possibility to bi-coherent states, and we discuss in many details an example based on a couple of superpotentials first introduced in \cite{bag2010jmp}. The results deduced here belong to the same distributional approach to pseudo-bosons first proposed in \cite{bag2020JPA}.

\end{abstract}

\vspace{2cm}


\vfill


\newpage

\section{Introduction}

It is today out of doubt that non self-adjoint Hamiltonians describe relevant physical systems: gain-loss systems, or systems having $PT$-symmetry are often connected with observables which, at least apparently, do not satisfy the {\em standard} self-adjointness condition: if $\Sc$ is such a system, described in a certain Hilbert space $\Hil$, with scalar product $\langle .,.\rangle$, its {\em natural} adjoint is defined in term of the scalar product: given $X\in B(\Hil)$, the set of bounded operators on $\Hil$\footnote{We restrict here to bounded operators for simplicity. The adjoint can be defined also for unbounded operators, adding some conditions on the domains.}, its adjoint $X^\dagger$ is defined by $\langle X^\dagger f,g\rangle=\langle f,Xg\rangle$, $\forall f,g\in\Hil$. A self-adjoint operator $H$ satisfies the equality $H=H^\dagger$. Since the seminal paper \cite{ben1} it became clear that it is possible to have a {\em meaningful physical models} also in presence of non self-adjoint Hamiltonians, at least if these Hamiltonians satisfy some different, and surely more physically motivated, symmetry conditions. Since then the interest in this {\em extended} quantum mechanics diffused among physicists, and produced several results, both from a theoretical and from an experimental point of view. Also, this interest was often shared by mathematicians, who started to contribute to the subject. What follows is a very partial (and biased) list of references, mainly focused on theoretical and mathematical aspects of this research: \cite{benbook}-\cite{petr3}.

From a mathematical side, loosing self-adjointness of the Hamiltonian $H$ immediately implies that its eigenvalues could be complex and its eigenvectors can be not mutually orthogonal. While the first aspect can be easily controlled, the second is more subtle. More explicitly, simple conditions ensuring that, even if $H\neq H^\dagger$, its eigenvalues are all real, are well known. For instance, it is sufficient that $H$ and $H^\dagger$ are similar, i.e., that an invertible operator $S$ exists such that $H=SH^\dagger S^{-1}$. It is important to stress once again that we are restricting here to bounded operators. Otherwise, the situation is not that simple, but many results exist also in this case, see \cite{bagbookPT}, for instance. As for the eigestates of $H$, these are no longer orthogonal, in general. However, a biorthogonal set can usually be found, made of the eigenstates of $H^\dagger$. However, even if the two sets of eigenvectors we find are usually complete (or total) in $\Hil$, they are not necessarily bases for $\Hil$. These aspects are discussed, for instance, in \cite{davies,baginbagbook}.

Suppose now that $H$ is factorizable, i.e. that two operators $A$ and $B$ exist, bounded or not, such that $H=BA$. This is what happens, in particular, for the very well known quantum harmonic oscillator, for which $H_0=c^\dagger c$, where $[c,c^\dagger]=\1$, the identity operator on $\Hil$. It is known that, in this case, $c$, $c^\dagger$ and $H_0$ are unbounded. Still there is a common domain of functions for all these (and other) operators.  $H_0$ has an orthonormal (o.n.) basis of eigenvectors, written in terms of Hermite polynomials.
Not surprisingly, in the very same way in which an o.n set of eigenvectors of, say, $H_0$ is {\em doubled} in two biorthogonal sets of eigenstates of $H$ and $H^\dagger$ when $c$ and $c^\dagger$ are replaced by $A$ and $B$, $[A,B]=\1$, the standard coherent state, eigenstate of $c$, is doubled into two different vectors, the so-called {\em bi-coherent states}, which are eigenstates of $A$ and $B^\dagger$ respectively.  This was discussed first in \cite{tri}, and then, with a slightly more mathematical taste, in \cite{bag2010}-\cite{bag2020}.

In all what discussed in the cited references, the role of the Hilbert space $\Hil$ was essential: operators are defined on $\Hil$ (or on some subspace of $\Hil$) and the vectors all belong to $\Hil$. However, it is known that there exist alternative approaches to quantum mechanics, particularly relevant when unbounded operators are involved. This is, for instance, the approach based on {\em rigged Hilbert spaces}, see for instance \cite{rigged1}-\cite{rigged4}. Another possibility, which was recently introduced in a context which is relevant for this paper, makes use of distribution theory, \cite{bgst} and \cite{bag2020JPA}. The essential remark, in both these papers, is that in concrete and simple physical problems it can be useful to abandon Hilbert spaces and to work outside, say, some $\Lc^2(\mathbb{R})$, with or without a measure. In other words, the (analytic) problems are not really solved even changing the scalar product, and replacing $\Lc^2(\mathbb{R})$ with $\Lc^2(\mathbb{R},p(x)\,dx)$, for some specific measure $p(x)$. This is part of what we will see also here. In fact, this aspect will be discussed here in some details in connection with a pseudo-bosonic system first introduced in \cite{bag2010jmp}, and then considered in \cite{bag2011} and, more recently, in \cite{bag2020}. As we will see, this system still has something to teach, both in the context of weak pseudo-bosons, and in connection with coherent states.

 The paper is organized as follows: in the next section we will list few results and definitions on pseudo-bosons and bi-coherent states which are relevant for us. In Section \ref{sectIII} we will consider two pseudo-bosonic operators defined by means of two different functions known as {\em superpotentials}, and their connected number-like operators, and we will focus on their eigenstates, considering the case in which these are square-integrable and when they are not. We will try to stay as general as possible, not fixing the form of the superpotentials: in doing so, we will extend many of the results deduced in \cite{bag2020,bag2010jmp,bag2011}, which were found under very special choices of superpotentials. The same dual situation is also considered in Section \ref{sectweakbcs}: the first part of this section is devoted to bi-coherent states living in $\Lc^2(\mathbb{R})$. In the second part we will define, for the first time in our knowledge, what we call {\em weak} bi-coherent states using the language of distribution, i.e., as continuous functionals on some particular function space. Section \ref{sectconl} contains our conclusions.

\section{Preliminaries}

In this paper we will reconsider a pair of ladder operators $A$ and $B$, with $B\neq A^\dagger$, originally introduced in \cite{bag2010jmp}, and more recently considered in \cite{bag2020}, to discuss how the notion of pseudo-bosons, and bi-coherent states, can be extended outside $\Lc^2(\mathbb{R})$. This analysis continues what was done in \cite{bag2020JPA}, where our interest was motivated by a rigorous analysis of the position and the momentum operators. To keep the paper self-contained, we devote this section to list few useful definitions and results on pseudo-bosons and on bi-coherent states.

\subsection{$\D$-pseudo bosons: basic facts}\label{sectpbs}

Let $\Hil$ be a given Hilbert space with scalar product $\left<.,.\right>$ and related norm $\|.\|$. Let $a$ and $b$ be two operators
on $\Hil$, with domains $D(a)\subset \Hil$ and $D(b)\subset \Hil$ respectively, $a^\dagger$ and $b^\dagger$ their adjoint, and let $\D$ be a dense subspace of $\Hil$
such that $a^\sharp\D\subseteq\D$ and $b^\sharp\D\subseteq\D$. Here with $x^\sharp$ we indicate $x$ or $x^\dagger$. Of course, $\D\subseteq D(a^\sharp)$
and $\D\subseteq D(b^\sharp)$.

\begin{defn}\label{def21}
	The operators $(a,b)$ are $\D$-pseudo bosonic  if, for all $f\in\D$, we have
	\be
	a\,b\,f-b\,a\,f=f.
	\label{A1}\en
\end{defn}

When $b=a^\dagger$, this is simply the canonical commutation relation (CCR) for ordinary bosons. However, when the CCR is replaced by (\ref{A1}), the situation changes, becoming mathematically more interesting. In particular, it is useful to assume the following:

\vspace{2mm}

{\bf Assumption $\D$-pb 1.--}  there exists a non-zero $\varphi_{ 0}\in\D$ such that $a\,\varphi_{ 0}=0$.

\vspace{1mm}

{\bf Assumption $\D$-pb 2.--}  there exists a non-zero $\Psi_{ 0}\in\D$ such that $b^\dagger\,\Psi_{ 0}=0$.

\vspace{2mm}

It is obvious that, since $\D$ is stable under the action of $b$ and $a^\dagger$, then  $\varphi_0\in D^\infty(b):=\cap_{k\geq0}D(b^k)$ and  $\Psi_0\in D^\infty(a^\dagger)$, so
that the vectors \be \varphi_n:=\frac{1}{\sqrt{n!}}\,b^n\varphi_0,\qquad \Psi_n:=\frac{1}{\sqrt{n!}}\,{a^\dagger}^n\Psi_0, \label{A2}\en
$n\geq0$, can be defined and they all belong to $\D$. Hence, they also belong to the domains of $a^\sharp$, $b^\sharp$ and $N^\sharp$, where $N=ba$. We see that, from a practical point of view, $\D$ is the natural space to work with and, in this sense, it is even more relevant than $\Hil$. $\D$ is a sort of {\em physical vector space}, endowed with a scalar product and with its related norm, but $\D$ is not complete (in general) with respect to this norm.  Let's put $\F_\Psi=\{\Psi_{ n}, \,n\geq0\}$ and
$\F_\varphi=\{\varphi_{ n}, \,n\geq0\}$.
It is  simple to deduce the following lowering and raising relations:
\be
\left\{
\begin{array}{ll}
	b\,\varphi_n=\sqrt{n+1}\varphi_{n+1}, \qquad\qquad\quad\,\, n\geq 0,\\
	a\,\varphi_0=0,\quad a\varphi_n=\sqrt{n}\,\varphi_{n-1}, \qquad\,\, n\geq 1,\\
	a^\dagger\Psi_n=\sqrt{n+1}\Psi_{n+1}, \qquad\qquad\quad\, n\geq 0,\\
	b^\dagger\Psi_0=0,\quad b^\dagger\Psi_n=\sqrt{n}\,\Psi_{n-1}, \qquad n\geq 1,\\
\end{array}
\right.
\label{A3}\en as well as the eigenvalue equations $N\varphi_n=n\varphi_n$ and  $N^\dagger\Psi_n=n\Psi_n$, $n\geq0$. If  $\left<\varphi_0,\Psi_0\right>=1$, then
\be \left<\varphi_n,\Psi_m\right>=\delta_{n,m}, \label{A4}\en
for all $n, m\geq0$. Hence $\F_\Psi$ and $\F_\varphi$ are biorthonormal. It is easy to see that, if $b=a^\dagger$, then $\varphi_n=\Psi_n$, so that biorthogonality is replaced by a simpler orthonormality. Moreover, the relations in (\ref{A3}) collapse, and only one number operator exists, since in this case $N=N^\dagger$.

The analogy with ordinary bosons suggests us to consider the following:

\vspace{2mm}

{\bf Assumption $\D$-pb 3.--}  $\F_\varphi$ is a basis for $\Hil$.

\vspace{1mm}

This is equivalent to requiring that $\F_\Psi$ is a basis for $\Hil$ as well, \cite{chri}. However, several  physical models show that $\F_\varphi$ is {\bf not} always a basis for $\Hil$, but it is still complete in $\Hil$: if $f\in\Hil$ is orthogonal to $\varphi_n$, for all $n$, then $f=0$. For this reason we adopt the following weaker version of  Assumption $\D$-pb 3, \cite{baginbagbook}:

\vspace{2mm}

{\bf Assumption $\D$-pbw 3.--}  For some subspace $\G$ dense in $\Hil$, $\F_\varphi$ and $\F_\Psi$ are $\G$-quasi bases.

\vspace{2mm}
This means that, for all $f$ and $g$ in $\G$,
\be
\left<f,g\right>=\sum_{n\geq0}\left<f,\varphi_n\right>\left<\Psi_n,g\right>=\sum_{n\geq0}\left<f,\Psi_n\right>\left<\varphi_n,g\right>,
\label{A4b}
\en
which can be seen as a weak form of the resolution of the identity, restricted to $\G$. Of course, if $f\in\G$ is orthogonal to all the $\varphi_n$'s, or to all the $\Psi_n$'s, then (\ref{A4b}) implies that $f=0$. Hence $\F_\varphi$ and $\F_\Psi$ are complete in $\G$, \cite{bagbell}.

The families $\F_\varphi$ and $\F_\Psi$ can be used to define two densely defined operators $S_\varphi$ and $S_\Psi$ via their
action respectively on  $\F_\Psi$ and $\F_\varphi$: \be
S_\varphi\Psi_{ n}=\varphi_{ n},\qquad
S_\Psi\varphi_{ n}=\Psi_{\bf n}, \label{213}\en for all $ n$. These operators play a very import role in the analysis of pseudo-bosons, since they map $\F_\varphi$ into $\F_\Psi$ and vice-versa, and define new scalar products in $\Hil$ is terms of which, for instance, the (new) adjoint of $b$ turns out to coincide with $a$. These, and many other aspects which are not relevant here, can be found in \cite{bag2017,bag2020JPA,baginbagbook}. The role of operators like the $S_\varphi$ and $S_\Psi$ in (\ref{213}) in different, but related situations is discussed, for instance, in \cite{benbook,mosta}.

\subsection{Bi-coherent states}\label{sectbcs}

Concerning coherent states, rather than considering the strict $\D$-pseudo bosonic situation described before, where $a$, $b$, $\varphi_n$ and $\Psi_n$ satisfy the ladder equations in (\ref{A3}), we give here a more general result which is again based on certain biorthogonal sets, and some ladder operators, but not necessarily pseudo-bosonic. In particular, as we will show in a moment, our results here hold when replacing $\sqrt{n}$ in the second and in the fourth rows in (\ref{A3}) with some $\alpha_n$, see below. This approach is convenient since bi-coherent states for pseudo-bosons can be seen as a special case of a much more general situation, described by (\ref{20}) below and by Theorem \ref{theo1}. 

Let us consider two biorthogonal families of vectors, $\F_{\tilde\varphi}=\{\tilde\varphi_n\in\Hil, \, n\geq0\}$ and $\F_{\tilde\Psi}=\{\tilde\Psi_n\in\Hil, \, n\geq0\}$ which are $\G$
-quasi bases for some dense subset of $\Hil$, see (\ref{A4b}). Consider an increasing sequence of real numbers $\alpha_n$ satisfying the  inequalities $0=\alpha_0<\alpha_1<\alpha_2<\ldots$. We call $\overline\alpha$ the limit of $\alpha_n$ for $n$ diverging, which coincides with $\sup_n\alpha_n$. We further consider two operators, $A$ and $B^\dagger$, which act as lowering operators respectively on $\F_{\tilde\varphi}$ and $\F_{\tilde\Psi}$ in the following way:
\be
A\,\tilde\varphi_n=\alpha_n\tilde\varphi_{n-1}, \qquad B^\dagger\,\tilde\Psi_n=\alpha_n\tilde\Psi_{n-1},
\label{20}\en
for all $n\geq1$, with $A\,\tilde\varphi_0=B^\dagger\,\tilde\Psi_0=0$. These are the lowering equations which replace those in (\ref{A3}). Then the following theorem holds, \cite{bagproc}:

\begin{thm}\label{theo1}
	Assume that four strictly positive constants $A_\varphi$, $A_\Psi$, $r_\varphi$ and $r_\Psi$ exist, together with two strictly positive sequences $M_n(\varphi)$ and $M_n(\Psi)$, for which
	\be
	\lim_{n\rightarrow\infty}\frac{M_n(\varphi)}{M_{n+1}(\varphi)}=M(\varphi), \qquad \lim_{n\rightarrow\infty}\frac{M_n(\Psi)}{M_{n+1}(\Psi)}=M(\Psi),
	\label{21}\en
	where $M(\varphi)$ and $M(\Psi)$ could be infinity, and such that, for all $n\geq0$,
	\be
	\|\tilde\varphi_n\|\leq A_\varphi\,r_\varphi^n M_n(\varphi), \qquad \|\tilde\Psi_n\|\leq A_\Psi\,r_\Psi^n M_n(\Psi).
	\label{22}\en
	Then, putting $\alpha_0!=1$ and $\alpha_k!=\alpha_1\alpha_2\cdots\alpha_k$, $k\geq1$, the following series:
	\be
	N(|z|)=\left(\sum_{k=0}^\infty\frac{|z|^{2k}}{(\alpha_k!)^2}\right)^{-1/2},
	\label{23}\en
	\be
	\varphi(z)=N(|z|)\sum_{k=0}^\infty\frac{z^k}{\alpha_k!}\tilde\varphi_k,\qquad \Psi(z)=N(|z|)\sum_{k=0}^\infty\frac{z^k}{\alpha_k!}\tilde\Psi_k,
	\label{24}\en
	are all convergent inside the circle $C_\rho(0)$ in $\mathbb{C}$ centered in the origin of the complex plane and of radius $\rho=\overline\alpha\,\min\left(1,\frac{M(\varphi)}{r_\varphi},\frac{M(\Psi)}{r_\Psi}\right)$. Moreover, for all $z\in C_\rho(0)$,
	\be
	A\varphi(z)=z\varphi(z), \qquad B^\dagger \Psi(z)=z\Psi(z).
	\label{25}\en
	Suppose further that a measure $d\lambda(r)$ does exist such that
	\be
	\int_0^\rho d\lambda(r) \,r^{2k}=\frac{(\alpha_k!)^2}{2\pi},
	\label{26}\en
	for all $k\geq0$. Then, putting $z=re^{i\theta}$ and calling $d\nu(z,\overline z)=N(r)^{-2}d\lambda(r)d\theta$, we have
	\be
	\int_{C_\rho(0)}\left<f,\Psi(z)\right>\left<\varphi(z),g\right>d\nu(z,\overline z)=
	\int_{C_\rho(0)}\left<f,\varphi(z)\right>\left<\Psi(z),g\right>d\nu(z,\overline z)=
	\left<f,g\right>,
	\label{27}\en
	for all $f,g\in\G$.
	
\end{thm}

Some comments are in order: first we observe that, contrarily to what happens for ordinary coherent states, \cite{aagbook,didier,gazeaubook},  the norms of the vectors $\tilde\varphi_n$ and $\tilde\Psi_n$ need not being uniformly bounded, here. On the contrary, they can diverge rather fast with $n$. Of course, this is reflected by the fact that bi-coherent states of this kind only exist inside $C_\rho(0)$. Secondly, we are here fully working in an Hilbert space. This is clear because, for each fixed $n$, the norms of both $\tilde\varphi_n$ and $\tilde\Psi_n$ are bounded, as well as the norms in $\Hil$  of $\varphi(z)$ and $\Psi(z)$. However, and this is part of our new results in this paper, a similar strategy can be extended also to the case in which $\tilde\varphi_n$ or $\tilde\Psi_n$, or both, do not belong to $\Hil$, so that $\|\tilde\varphi_n\|=\|\tilde\Psi_n\|=\infty$, for all (or some) $n$. The analysis of this situation, motivated by what discussed in Section \ref{sectIII}, is contained in Section \ref{sectweakbcs}.
Finally, if $A$ and $B$ satisfy Definition \ref{def21}, then $\alpha_n=\sqrt{n}$ and $\overline\alpha=\rho=\infty$. Hence we have convergence of the series above in all the complex plane, and the measure $d\lambda(r)$ can be easily identified. We will return on this aspect later.

\section{A class of choices}\label{sectIII}

The Hilbert space we work with, here, is $\Hil=\Lc^2(\Bbb R)$ with the usual scalar product, and we consider, \cite{bag2010jmp,bag2011,bag2020}, the operators
\be
A=\frac{d}{dx}+w_A(x), \qquad B=-\frac{d}{dx}+w_B(x).
\label{31}\en
where $w_A(x)$ and $w_B(x)$ are two $C^\infty$ functions such that $w_A(x)\neq \overline{w_B(x)}$, in order for $B^\dagger$ to be different from $A$. The reason for asking this regularity to $w_A(x)$ and $w_B(x)$ is because it will be used in the following, in deriving some useful result. However, most of the times, one could require much less regularity, as one can see in what follows. Adopting the terminology used in \cite{bag2020}, we call these functions superpotentials, the reason being that they can be used, in general to construct two supersymmetric Hamiltonians $H_1=BA$ and $H_2=AB$ whose eigenvectors are related as in usual supersymmetric quantum mechanics, \cite{CKS,jun,bag2020}. The fact that $H_1$ and $H_2$ are supersymmetric partners was considered in details in \cite{bag2020}. In this paper this aspect is not particularly relevant, since with our constraint on the superpotentials, we have $H_2=H_1+\1$. This is due to the commutation rule $[A,B]=\1$ which we are going to assume on $A$ and $B$, following Definition \ref{def21}. Moreover, as we will show in a moment, these two superpotentials must be connected. Now, since
\be
H_1=BA=-\frac{d^2}{dx^2}+q_1(x)\frac{d}{dx}+V_1(x), \qquad H_2=AB=-\frac{d^2}{dx^2}+q_1(x)\frac{d}{dx}+V_2(x),
\label{32}\en
where
\be
q_1(x)=w_B(x)-w_A(x), \quad V_1(x)=w_A(x)w_B(x)-w'_A(x), \quad V_2(x)=w_A(x)w_B(x)+w'_B(x),
\label{33}\en  
we find that
\be
[A,B]=H_2-H_1=V_2(x)-V_1(x)=w_A'(x)+w_B'(x).
\label{34}\en
Hence, to have $[A,B]=\1$, which is the case which is interesting for us here, we must have $w_A(x)+w_B(x)=x+k$, for a generic $k$ which we take to be real from now on: $k\in\mathbb{R}$. This does not prevent the superpotentials to be complex valued, in general. We will call these functions {\em pseudo-bosonic superpotentials} (PBSs). It is worth noticing that the commutators here must be understood in the sense of Definition \ref{def21}, since they involve unbounded operators for which commutators is mathematically a {\em risky operation}. 

It is clear that the vacua of $A$ and $B^\dagger$, see Section \ref{sectpbs}, are the following:
\be
\varphi_0(x)=N_{\varphi}\,\exp\left\{-s_A(x)\right\}, \qquad \Psi_0(x)=N_{\Psi}\,\exp\left\{-\overline{s_B(x)}\right\},
\label{35}\en
where 
$s_A(x)=\int w_A(x)dx$ and $s_B(x)=\int w_B(x)dx$, and
$N_{\varphi}$ and $N_{\Psi}$ are two normalization constants which we will compute soon. Indeed a direct check shows that $A\varphi_0(x)=B^\dagger\Psi_0(x)=0$.

It is useful to stress once more that formulas in (\ref{31}) somehow extend the standard definitions in SUSY quantum mechanics, see \cite{CKS,jun} for instance, where a single superpotential $w(x)$ is introduced, and found solving a Riccati equation, in the attempt to factorize a given, self-adjoint, Hamiltonian. As already observed, while this aspect can be useful in general, when $[A,B]$ has not a simple expression, it is not so relevant here exactly because this commutator is simple and, as such, informations on $H_1$ can be easily translated to analogous informations on $H_2$.

What is interesting for us is to extend now our previous analysis in \cite{bag2010jmp,bag2011,bag2020}. In particular, we are not interested in requiring anymore that  $\varphi_{0}(x)$ and $\Psi_0(x)$ are both in $\Lc^2(\mathbb{R})$ or, even more, in some dense subspace $\D$ of $\Lc^2(\mathbb{R})$, as Assumptions $\D$-pb 1 and $\D$-pb2  in Section \ref{sectpbs}. We rather want to see how much of the results in Section \ref{sectpbs} can be deduced also weakening these requirements. This is exactly on the same line of our recent results on weak pseudo-bosons, \cite{bag2020JPA}, where Hilbert spaces were not so important and were replaced by other functional spaces, as those of distributions $\D'(\mathbb{R})$ or of tempered distributions $\Sc'(\mathbb{R})$. To be concrete, in this paper we will not assume, except when explicitly stated, that $\varphi_{0}(x), \Psi_0(x)\in\Lc^2(\mathbb{R})$. Nevertheless, we have the following result:

\begin{lemma}\label{lemma1}
	If $w_A(x)$ and $w_B(x)$ are $C^\infty$ PBSs, then $\varphi_{0}(x)\, \overline{\Psi_0(x)}\in\Lc^1(\mathbb{R})$. Moreover, if $N_\varphi\overline{N_\Psi}=\frac{e^{-k^2/4}}{\sqrt{2\pi}}$, then $\langle\Psi_0,\varphi_0\rangle=1$.
\end{lemma}

\begin{proof}
	Indeed we have, recalling that $w_A(x)+w_B(x)=x+k$,
	$$
	\varphi_{0}(x)\, \overline{\Psi_0(x)}=N_\varphi\overline{N_\Psi}e^{-(s_A(x)+s_B(x))}=N_\varphi\overline{N_\Psi}e^{-\frac{x^2}{2}-kx-q},
	$$
	where $q$ is an integration constant which, from now on, we put equal to zero, reabsorbing its effect in the normalization constants $N_\varphi$ and $N_\Psi$. Hence $\varphi_{0}(x)\, \overline{\Psi_0(x)}$ is a gaussian and, as such, is in $\Lc^1(\mathbb{R})$. The rest of the Lemma follows from an explicit simple computation.	
	
\end{proof}

 The result of this Lemma depicts what happens, in general, for the so-called PIP-spaces\footnote{Here PIP stands for partial inner product.}, \cite{anttra}, where the scalar product in $\Lc^2(\mathbb{R})$ is extended to a sesquilinear form over {\em compatible spaces}: roughly speaking, $V_1$ and $V_2$ are compatible in $\Hil$ if, taken $f_1\in V_1$ and $f_2\in V_2$,  a sesquilinear form $\langle f_1,f_2\rangle$ which extends the ordinary scalar product on $\Hil$, which we also indicate as $\langle.,.\rangle$, makes sense.  Then Lemma \ref{lemma1} states that $\varphi_0(x)$ and $\Psi_0(x)$ belong to compatible spaces, for all possible choices of PBSs.
 
 \vspace{2mm}
 
 {\bf Remarks:--} (1) In Lemma \ref{lemma1} it is not important that $w_A(x)$ and $w_B(x)$ are $C^\infty$ functions. It is sufficient that they can be integrated, so that $s_A(x)$ and $s_B(x)$ both exist. We prefer to introduce this assumption already here because it will be quite relevant in the following.
 
 (2) According to the Lemma, from now on we will assume that
 \be N_\varphi\overline{N_\Psi}=\frac{e^{-k^2/4}}{\sqrt{2\pi}},\label{35bis}\en
 in order to have $\langle\Psi_0,\varphi_0\rangle=1$. Moreover, with our choice of $q$ in the proof of Lemma \ref{lemma1}, we fix
 \be s_A(x)+s_B(x)=\frac{x^2}{2}+kx,
 \label{35tris}\en
 all throughout the paper.
 \vspace{2mm}

 Let us now introduce, in analogy with (\ref{A2}), the following vectors: 
 	\be
 \varphi_n(x)=\frac{1}{\sqrt{n!}}B^n\varphi_0(x), \qquad \psi_n(x)=\frac{1}{\sqrt{n!}}{A^\dagger}^n\psi_0(x),
 \label{36}\en
 $n=1,2,3,\ldots$. It is clear that, since there is no guarantee that $\varphi_0(x)$ and $\Psi_0(x)$ are square-integrable, we are not in the situation described in Section \ref{sectpbs}: $\varphi_0(x)$ and $\Psi_0(x)$ are not, except that for special choices of PBSs, in any suitable $\D$ as required in Assumptions $\D$-pb 1. and $\D$-pb 2. Nevertheless, $\varphi_n(x)$ and $\Psi_n(x)$ can be computed and turn out to be $C^\infty$ functions for all possible choices of $C^\infty$ PBSs. This is the content of the following Proposition, which generalizes a similar result deduced in \cite{bag2020} for specific forms of $w_A(x)$ and $w_B(x)$, and in which the role of $\Lc^2(\mathbb{R})$ was essential.

\begin{prop}
	For any choice of $C^\infty$ PBSs we have
	\be
	\frac{\varphi_n(x)}{\varphi_0(x)}=\frac{\psi_n(x)}{\psi_0(x)}=p_n(x,k),
	\label{37}\en
	for all $n\geq0$, where $p_n(x,k)$ is independent of $w_A(x)$ and $w_B(x)$ and is defined recursively as follows:
	\be
	p_0(x,k)=1, \qquad p_{n}(x,k)=\frac{1}{\sqrt{n}}\left(p_{n-1}(x,k)(x+k)-p_{n-1}'(x,k)\right),
	\label{38}\en
	$n\geq1$. Moreover, for all $n\geq0$,
	\be
	p_n(x,k)=\frac{1}{\sqrt{2^n \,n!}}\,H_n\left(\frac{x+k}{\sqrt{2}}\right).
	\label{39}\en
	
\end{prop}

\begin{proof}
We start proving (\ref{37}) for the $\varphi_n(x)$. The proof for the $\Psi_n(x)$ is completely analogous.

We use induction on $n$. It is clear that (\ref{37}) is true if $n=0$, since $p_0(x,k)=1$. Suppose now that the statement is true for $n-1$, so that
$
\varphi_{n-1}(x)=p_{n-1}(x,k)\,\varphi_0(x),
$
with $p_{n-1}(x,k)$ satisfying (\ref{38}), for $n$ replaced by $n-1$. Let us then compute $\varphi_n(x)$. From (\ref{36}) we have
$$
\sqrt{n}\,\varphi_n(x)=B\varphi_{n-1}(x)=\left(-\frac{d}{dx}+w_B(x)\right)\left(p_{n-1}(x,k)\,\varphi_0(x)\right).
$$
Now, recalling that $\frac{d}{dx}\varphi_0(x)=-w_A(x)\varphi_0(x)$ and that $w_A(x)+w_B(x)=x+k$, we get the following:
$$
\sqrt{n}\,\varphi_n(x)=\left(-p_{n-1}'(x,k)+(x+k)p_{n-1}(x,k)\right)\varphi_0(x),
$$
from which our first claim easily follows. Equality (\ref{39}) is a consequence of the Rodriguez formula for the Hermite polynomials,
$$
H_n(y)=(-1)^ne^{y^2}\frac{d^n}{dy^n}\,e^{-y^2},
$$
and from the fact that, again using induction on $n$, we can check that
$$
p_n(x,k)=(-1)^n\frac{1}{\sqrt{n!}}\,e^{x^2/2+kx}\frac{d^n}{dx^n}\,e^{-x^2/2-kx}.
$$

\end{proof}

Summarizing we have
\be
\varphi_n(x)=\frac{N_\varphi}{\sqrt{2^n\,n!}}\,H_n\left(\frac{x+k}{\sqrt{2}}\right)\,e^{-s_A(x)},
\label{310}\en
and
\be
\Psi_n(x)=\frac{N_\Psi}{\sqrt{2^n\,n!}}\,H_n\left(\frac{x+k}{\sqrt{2}}\right)\,e^{-\overline{s_B(x)}}.
\label{311}\en
In analogy with what discussed in Section \ref{sectpbs},  these functions satisfy the following eigenvalue equations:
\be
N\varphi_n(x)=n\varphi_n(x), \qquad N^\dagger\Psi_n(x)=n\Psi_n(x),
\label{311bis}\en
where $N=BA$ and $N^\dagger=A^\dagger B^\dagger$.

It is clear that, without further assumptions on the PBSs, these functions (or part of these functions) could be not square-integrable. However, Lemma \ref{lemma1} can be extended as follows

\begin{prop}\label{prop1}
	If $w_A(x)$ and $w_B(x)$ are $C^\infty$ PBSs, then $\varphi_{n}(x)\, \overline{\Psi_m(x)}\in\Lc^1(\mathbb{R})$ and  $\langle\Psi_m,\varphi_n\rangle=\delta_{n,m}$, for all $n,m\geq0$.
\end{prop}

\begin{proof}
Indeed we have, in analogy with Lemma \ref{lemma1},
	$$
	\varphi_{n}(x)\, \overline{\Psi_m(x)}=\frac{N_\varphi\,\overline{N_\Psi}}{\sqrt{2^{n+m}\,n!\,m!}}\,H_n\left(\frac{x+k}{\sqrt{2}}\right)H_m\left(\frac{x+k}{\sqrt{2}}\right)\,e^{-\frac{x^2}{2}-kx}
	$$
	Hence, putting $y=\frac{x+k}{\sqrt{2}}$ in $\langle\Psi_m,\varphi_n\rangle$ and using (\ref{35bis}), together with the integral
	$$
	\int_{\Bbb{R}}H_n(x)H_m(x)e^{-x^2}\,dx=\sqrt{\pi}\,2^n\,n! \delta_{n,m},
	$$
	our claim follows.	
	
\end{proof}

Incidentally we observe that this approach significantly simplifies the proof given in \cite{bag2020}, which furthermore was deduced for a very specific choice of PBSs, $w_A(x)=e^x+k$ and $w_B(x)=x-e^x$. Also, it is useful to stress that this Proposition produces, using explicit computations, the result in (\ref{A4}), which however was deduced under much stronger assumptions on $\F_\varphi=\left\{\varphi_n(x),\,n\geq0\right\}$ and $\F_\Psi=\left\{\Psi_n(x),\,n\geq0\right\}$.

Using the same language adopted after Lemma \ref{lemma1}, Proposition \ref{prop1} shows that $\varphi_n(x)$ and $\Psi_m(x)$ belong to compatible spaces so that, even if one of the two functions is not square-integrable, still the ordinary scalar product in $\Lc^2(\mathbb{R})$ can be extended to any product of these functions.

In principle, without extra assumptions on the PBSs, we don't know if the functions of $\F_\varphi$ and $\F_\Psi$ are square-integrable or not. For this reason, it makes not much sense to check if they are biorthogonal bases, or if they are complete, in  $\Lc^2(\mathbb{R})$ or not, see also \cite{bag2020JPA}. However, as in \cite{bag2020JPA}, something can still be deduced. To do so, we start rewriting $\varphi_n(x)$ and $\Psi_n(x)$ as follows:
\be
\varphi_n(x)=c_n(x)\rho_A(x), \qquad \Psi_n(x)=c_n(x)\rho_B(x),
\label{312}\en
where
\be
c_n(x)=\frac{1}{2^{1/4}}\,e_n\left(\frac{x+k}{\sqrt{2}}\right), \qquad \mbox{where }\quad  e_n(x)=\frac{1}{\sqrt{2^nn!\sqrt{\pi}}}H_n(x)\,e^{-\frac{x^2}{2}},
\label{313}\en
is the well known $n$-th eigenstate of the harmonic oscillator, and
\be
 \rho_A(x)=N_\varphi(2\pi)^{1/4}e^{\frac{1}{2}\left(\frac{x+k}{\sqrt{2}}\right)^2}e^{-s_A(x)},\qquad \rho_B(x)=N_\Psi(2\pi)^{1/4}e^{\frac{1}{2}\left(\frac{x+k}{\sqrt{2}}\right)^2}e^{-\overline{s_B(x)}}.
\label{314}\en
Notice that these functions are independent of $n$. Because of the (\ref{35tris}) $\rho_A(x)$ and $\rho_B(x)$ are not mutually independent. Indeed they satisfy
\be
\rho_A(x)\,\overline{\rho_B(x)}=1,
\label{315}\en
so that one is (a part a complex conjugation) the inverse of the other. This is in agreement, of course, with the factorization in (\ref{312}) and with Proposition \ref{prop1}, since the set $\F_c=\{c_n(x), \, n\geq 0\}$ is an o.n. basis for $\Lc^2(\mathbb{R})$. Let us now define the following set:
\be\E=\left\{h(x)\in\Lc^2(\mathbb{R}):\, h(x)\rho_j(x)\in\Lc^2(\mathbb{R}), \, j=A,B\right\}
\label{316}\en
This is clearly a subset of $\Lc^2(\mathbb{R})$, and is dense in it, since it contains the set $\D(\mathbb{R})$ of the $C^\infty$ compactly supported functions, which is dense in $\Lc^2(\mathbb{R})$. The inclusion $\D(\mathbb{R})\subseteq\E$ can be easily checked: take $d(x)\in\D(\mathbb{R})$. Then, since $\rho_A(x)$ and $\rho_B(x)$ are $C^\infty$ functions, $d(x)\rho_j(x)\in\D(\mathbb{R})$, $j=A,B$. Hence $d(x)\in\E$, since $\D(\mathbb{R})\subset\Lc^2(\mathbb{R})$. Now, it is possible to prove, using the same terminology as in (\ref{A4b}), that $(\F_\varphi,\F_\Psi)$ are $\E$-quasi bases. More explicitly, recalling that not all the functions considered here are in $\Lc^2(\mathbb{R})$, we need to check that,
for all $\epsilon(x)\in\E$,  $\langle \epsilon,\varphi_n\rangle$ and $\langle \epsilon,\Psi_n\rangle$ are well defined for all $n\geq0$, and that, for all $f(x),g(x)\in\E$,
 \be
 \left<f,g\right>=\sum_{n\geq0}\langle f,\varphi_n\rangle\langle\Psi_n,g\rangle=\sum_{n\geq0}\left<f,\Psi_n\right>\left<\varphi_n,g\right>,
 \label{317}
 \en
The proof of these claims is based on the definition of $\E$ and on the fact that $\F_c$ is an o.n. basis for $\Lc^2(\mathbb{R})$. Indeed using (\ref{312}) we have, for instance,
$$
\langle \epsilon,\varphi_n\rangle=\int_{\Bbb{R}}\overline{\epsilon(x)}\,\varphi_n(x)\,dx=\int_{\Bbb{R}}\overline{\overline{\rho_A(x)}\,\epsilon(x)}\,c_n(x)\,dx,
$$
which is well defined for all $n$ since $\overline{\rho_A(x)}\,\epsilon(x)\in\Lc^2(\mathbb{R})$. Similarly, $$\langle \epsilon,\Psi_n\rangle=\int_{\Bbb{R}}\overline{\overline{\rho_B(x)}\,\epsilon(x)}\,c_n(x)\,dx,$$ 
which is also well defined for all $n$, since $\overline{\rho_B(x)}\,\epsilon(x)\in\Lc^2(\mathbb{R})$ as well. To check (\ref{317}), we now use completeness of $\F_c$ in $\Lc^2(\mathbb{R})$:
$$
\sum_{n\geq0}\langle f,\varphi_n\rangle\langle\Psi_n,g\rangle=\sum_{n\geq0}\langle \overline{\rho_A}\,f,c_n\rangle\langle c_n,\overline{\rho_B}\,g\rangle=\langle \overline{\rho_A}\,f,\overline{\rho_B}\,g\rangle=\int_{\Bbb{R}}\overline{f(x)}\,g(x)\,\rho_A(x)\overline{\rho_B(x)}=\left<f,g\right>,
$$
because of the (\ref{315}). The other part of (\ref{317}) can be proved similarly. We can conclude that, even if not all the functions considered in our analysis necessarily belong to $\Lc^2(\mathbb{R})$, still we can say that $(\F_\varphi,\F_\Psi)$ are $\E$-quasi bases, even if in a generalized sense. This is because $\E$ is {\em so well behaved} to take care of the bad asymptotic behaviour of the $\varphi_n(x)$ and $\Psi_n(x)$, when needed\footnote{In fact, for special choices of PBSs, it might still happen that these functions are all in $\Lc^2(\mathbb{R})$, \cite{bag2010jmp,bag2011}.In this case, $(\F_\varphi,\F_\Psi)$ are $\E$-quasi bases in the standard sense, \cite{baginbagbook}.}. This is very close to what happens in distributions theory, where the {\em bad analytic properties} of distributions are somehow cured by very well behaved test functions. We will see something similar to this in Section \ref{sectweakbcs}. 

Since $(\F_\varphi,\F_\Psi)$ are $\E$-quasi bases, it easily follows that both $\F_\varphi$ and $\F_\Psi$ are complete in $\E$, in the sense of \cite{bagbell}: if $\epsilon(x)\in\E$ is orthogonal to all the $\varphi_n(x)$, or to all the $\Psi_n(x)$, then $\epsilon(x)=0$ almost everywhere (a.e.)  in $\mathbb{R}$. To check this, it is enough to take $f(x)=g(x)$ in (\ref{317}) and assume that, say, $\langle f,\varphi_n\rangle=0$ for all $n$. Hence $\|f\|^2=0$, so that $f=0$. We recall, see \cite{bagbell}, that this does not imply that $\F_\varphi$ or $\F_\Psi$ are also complete in $\Lc^2(\mathbb{R})$. To check completeness of these sets in $\Lc^2(\mathbb{R})$, or in some even larger set, we recall first the following useful result, \cite{kolfom}.

\begin{lemma}
	 Suppose $\eta(x)$ is a Lebesgue-measurable function which is different from zero a.e. in $\mathbb{R}$. Suppose further that there exist two positive constants $\delta, C$ such that $|\eta(x)|\leq C\,e^{-\delta|x|}$ a.e. in $\mathbb{R}$, then the set $\left\{x^n\,\eta(x)\right\}$ is complete in $\Lc^2(\mathbb{R})$.
\end{lemma} 

Of course, a similar result holds if the monomials $\{x^n\}$ are replaced by other polynomials of the $n$-th degree, like the shifted and dilated version of the Hermite polynomial $H_n\left(\frac{x+k}{\sqrt{2}}\right)$ appearing in (\ref{313}).

In order to apply this Lemma, recalling (\ref{310}) and (\ref{311}) together with (\ref{35tris}), we conclude that $s_A(x)$ or $s_B(x)=\frac{x^2}{2}+kx-s_A(x)$, or both, must diverge to $+\infty$ at least as fast as $|x|$, for $|x|$ diverging. Of course, as also (\ref{315}) clearly shows, if $s_A(x)$ diverges too fast to $+\infty$, then $s_B(x)$ also diverges, but to $-\infty$ so that $e^{-\overline{s_B(x)}}$ cannot satisfy the main assumption of the Lemma above. Stated differently, if we want both $\F_\varphi$ and $\F_\Psi$ to be complete in  $\Lc^2(\mathbb{R})$, the PBSs must be chosen properly. Otherwise, only one of the two sets can be made of square-integrable functions.

Another completeness result follows from the following Theorem, \cite{walter}:

\begin{thm}
	Let $f(x)$ be a tempered distribution, $f(x)\in\Sc'(\mathbb{R})$, such that $\langle f,e_n\rangle=0$, for all $n\geq0$. Then $f(x)=0$.
\end{thm}
Here $e_n(x)$ has been introduced in (\ref{313}). This theorem can be restated by saying that the family $\{e_n(x)\}$ is complete not only in  $\Lc^2(\mathbb{R})$, but also in $\Sc'(\mathbb{R})$. Using this result it is now possible to prove that the sets $\F_\varphi$ and $\F_\Psi$ are complete in  $\hat\E$, defined as
$$
\hat\E=\left\{h(x)\in\Lc^2(\mathbb{R}):\, h(x)\rho_j(x)\in\Sc'(\mathbb{R}), \, j=A,B\right\}
$$
Notice that this set is much larger than $\E$, since $\Sc'(\mathbb{R})$ is obviously larger than $\Lc^2(\mathbb{R})$.

We refer to \cite{bag2010jmp} for more information on the operators in (\ref{31}) for choices of PBSs for which $\varphi_n(x), \Psi_n(x)\in\Lc^2(\mathbb{R})$. In the next section, we will analyse bi-coherent states mostly in the {\em worst case}, i.e. when a fully Hilbert space treatment of the system is not possible. In doing so, we will generalize some recent results on bi-coherent states, see Section \ref{sectbcs}, to their distributional versions.

\section{Weak bi-coherent states}\label{sectweakbcs}

As we have discussed in the previous section, if we don't choose properly the PBSs $w_A(x)$ and $w_B(x)$, we likely find eigenvectors of $N=BA$ ($\varphi_n(x)$, see (\ref{310})) and of $N^\dagger=A^\dagger B^\dagger$ ($\Psi_n(x)$, see (\ref{311})) which are not necessarily in $\Lc^2(\mathbb{R})$. Therefore, even if lowering equations like those in (\ref{20}) are satisfied, there is no reason why these vectors should satisfy the inequalities in (\ref{22}). In fact, it may happen that $\|\varphi_n\|$ or $\|\Psi_n\|$ are infinite. Of course, this is not the case if $\rho_A(x)$ and $\rho_B(x)$ belong to $\Lc^\infty(\mathbb{R})$. In fact, when this is true, from (\ref{312}) we deduce that
$$
\|\varphi_n\|=\|c_n\,\rho_A\|\leq\|\rho_A\|_\infty\|c_n\|=\|\rho_A\|_\infty, \qquad \|\Psi_n\|=\|c_n\,\rho_B\|\leq\|\rho_B\|_\infty,
$$
for all $n\geq0$. Moreover, comparing (\ref{20}) with (\ref{311bis}), we deduce that $\alpha_n=\sqrt{n}$, and therefore $\overline{\alpha}=\infty$. Hence Theorem \ref{theo1} holds with the following choice of the parameters involved: $A_\varphi=\|\rho_A\|_\infty$, $A_\Psi=\|\rho_B\|_\infty$, $r_\varphi=r_\Psi=M_n(\varphi)=M_n(\Psi)=1$, for all $n\geq0$. Hence $\rho=\infty$, and the series in (\ref{23}) and (\ref{24}) converge in all the complex plane. An example in which these conditions on $\rho_A(x)$ and $\rho_B(x)$ are satisfied is when we choose $s_A(x)$ and $s_B(x)$ as follows:
$$
s_A(x)=\frac{x^2}{4}+\frac{kx}{2}+\Phi(x), \qquad s_B(x)=\frac{x^2}{4}+\frac{kx}{2}-\Phi(x), 
$$
where $\Phi(x)$ is any real $C^\infty$ function bounded from below and from above, i.e. when there exist $m,M$ such that $-\infty<m\leq \Phi(x)\leq M<\infty$, a.e. in $\mathbb{R}$. In fact, in this case we can check that
$$
\|\rho_A\|_\infty=N_\varphi (2\pi)^{1/4}e^{k^2/4-m}, \qquad \|\rho_B\|_\infty=N_\Psi (2\pi)^{1/4}e^{k^2/4+M},
$$
while the PBSs are $w_A(x)=\frac{x}{2}+k+\Phi'(x)$ and $w_B(x)=\frac{x}{2}+k-\Phi'(x)$.
It is not difficult to find non trivial choices of $\Phi(x)$: $\cos(x)$ or $\sin(x)$, or some combinations of these, are just some examples. It may be interesting to notice that, if $\Phi(x)=0$, then $\varphi_n(x)$ coincides with $\Psi_n(x)$, except for a normalization constant, since both $\rho_A(x)$ and $\rho_B(x)$ turn out to be constant.

\subsection{An asymmetric example in $\Lc^2(\mathbb{R})$}

Let us take now $s_A(x)=\frac{x^2}{4}$ and $s_B(x)=\frac{x^2}{4}+kx$. Condition (\ref{35tris}) is satisfied and we have
\be
\varphi_n(x)=\frac{N_\varphi}{\sqrt{2^n\,n!}}\,H_n\left(\frac{x+k}{\sqrt{2}}\right)\,e^{-x^2/4},
\qquad
\Psi_n(x)=\frac{N_\Psi}{\sqrt{2^n\,n!}}\,H_n\left(\frac{x+k}{\sqrt{2}}\right)\,e^{-x^2/4-kx}.
\label{318}\en
Both these functions are square-integrable and the norm of $\varphi_n(x)$ can be computed using the following integral, \cite{grad}:
$$
\int_{\mathbb{R}}e^{-x^2}H_m(x+y)\,H_n(x+z)\,dx=2^n\,m!\sqrt{\pi}z^{n-m}L_m^{n-m}(-2yz),
$$
which holds if $m\leq n$. Here $L_m^{n-m}$ is a Laguerre polynomial. We get
\be\|\varphi_n\|^2=|N_\varphi|^2\sqrt{2\pi}\,L_n(-k^2),\label{319}\en
where $L_n(x)=L_n^0(x)$. As for the norm of $\Psi_n(x)$, the same formula above produces
\be\|\Psi_n\|^2=|N_\Psi|^2\sqrt{2\pi}\,e^{2k^2}\,L_n(-k^2),\label{320}\en
which differs from (\ref{319}) only for a constant factor. Now, to check if this example fits the assumptions of Theorem \ref{theo1}, we observe that the asymptotic behaviour (in $n$) of the Laguerre polynomials for negative arguments, as deduced in \cite{szego}, Theorem 8.22.3, is the following:
$$
L_n(x)=\frac{e^{x/2}2^{2(-nx)^{1/2}}}{2\sqrt{\pi}\,(-xn)^{1/4}}\left(1+O(n^{-1/2})\right),
$$
$x<0$, so that
\be
\|\varphi_n\|\simeq\frac{|N_\varphi|}{(2|k|)^{1/4}}\,e^{-k^2/4}\frac{e^{|k|\sqrt{n}}}{n^{1/8}}, \qquad 
\|\Psi_n\|\simeq\frac{|N_\Psi|}{(2|k|)^{1/4}}\,e^{3k^2/4}\frac{e^{|k|\sqrt{n}}}{n^{1/8}}, 
\label{321}\en
where $\simeq$ stands for {\em except for corrections $O(n^{-1/2})$}. Now, since we have clearly $e^{|k|\sqrt{n}}<e^{|k|n}$ for all non zero $k$ and $n\geq1$, we can identify the various ingredients of Theorem \ref{theo1} as follows:
$$
A_\varphi=\frac{|N_\varphi|}{(2|k|)^{1/4}}\,e^{-k^2/4}, \quad A_\Psi=\frac{|N_\Psi|}{(2|k|)^{1/4}}\,e^{3k^2/4}, \quad r_\varphi=r_\Psi=e^{|k|}, \quad M_n(\varphi)=M_n(\Psi)=\frac{1}{n^{1/8}}.
$$
Of course, this is valid if $k\neq0$, which is the only interesting case for us, since if $k=0$ then $s_A(x)=s_B(x)$. Hence $M(\varphi)=M(\Psi)=1$, $\rho=\infty$, and the bi-coherent states in (\ref{24}) are well defined in all the complex plane, belong to $\Lc^2(\mathbb{R})$, and satisfies the eigenvalue equations (\ref{24}) and the weak resolution of the identity (\ref{27}), with $\G=\E$. In fact, taking
\be d\lambda(r)=\frac{1}{\pi}e^{-r^2}r\,dr,
\label{322}\en
condition (\ref{26}) is satisfied:
$$
\int_0^\infty d\lambda(r) r^{2k}=\frac{k!}{2\pi}.
$$
Then (\ref{27}) follows from the fact that, as proved in Section \ref{sectIII}, $(\F_\varphi,\F_\Psi)$ are $\E$-quasi bases for all possible choices of PBSs satisfying (\ref{35tris}). Hence, in particular, they are $\E$-quasi bases for our choice $s_A(x)=\frac{x^2}{4}$ and $s_B(x)=\frac{x^2}{4}+kx$, or $w_A(x)=\frac{x}{2}$ and $w_B(x)=\frac{x}{2}+k$.

\subsection{Working outside $\Lc^2(\mathbb{R})$}\label{bcsexample}

Let us now concentrate on the situation in which $\Lc^2(\mathbb{R})$ is not the natural space where to work. Of course, what we have in mind is that changing the metric does not change much, leaving most of the problems unsolved. For instance, changing the scalar product can produce a convergent series for, say, $\varphi(z)$ but does help much when convergence of $\Psi(z)$ is also considered.  Or vice-versa. In other words, if we introduce a non trivial metric $\eta$ and a related scalar product $\langle.,.\rangle_\eta=\langle.,\eta.\rangle$, and the norm $\|.\|_\eta$, it may happen that $\|\varphi_n\|_\eta$ satisfies an inequality like the one in (\ref{22}). But $\|\Phi_n\|_\eta$ does not, in general. For this reason, in what follows we will describe a rather general version of the problem, and propose a solution for it. Then we will adopt the strategy we are going to construct to the situation considered in Section \ref{sectIII}.

Let $\F_c=\{c_n(x), \, n\geq0\}$ be an o.n. basis in $\Lc^2(\mathbb{R})$, not necessarily coincident with the basis in (\ref{313}), and let $\rho_f(x)$ and $\rho_g(x)$ be two Lebesgue-measurable functions such that, calling
\be
 f_n(x)=c_n(x)\,\rho_f(x), \qquad  g_n(x)=c_n(x)\,\rho_g(x),
\label{41}
\en
we have $f_n(x)\,g_m(x)\in\Lc^1(\mathbb{R})$, for all $n,m\geq0$. This implies that, despite of $f_n(x)$ or $g_n(x)$ being square-integrable or not, the form $\langle f_n,g_m\rangle$ is always well defined. With a slight abuse of language, we still call $\langle f_n,g_m\rangle$ the scalar product between $f_n(x)$ and $g_m(x)$. Motivated by (\ref{315}), we restrict here to the case in which $\rho_f(x)=\overline{\rho_g(x)}^{\,-1}$. In this way it is clear that $\F_f=\{f_n(x)\}$ and $\F_g=\{g_n(x)\}$ are biorthonormal: 
\be
\langle f_n,g_m\rangle=\delta_{n,m}.
\label{42}\en
We stress once more that we use this term, biorthonormal, in an extended sense, since we are not requiring here that $f_n(x)$ and $g_m(x)$ belong to $\Lc^2(\mathbb{R})$, while we are sure that $f_n(x)g_m(x)\in\Lc^1(\mathbb{R})$ anyhow. In analogy with (\ref{316}) we define now the set
\be\V=\left\{v(x)\in\Lc^2(\mathbb{R}):\, v(x)\rho_j(x)\in\Lc^2(\mathbb{R}), \, j=f,g\right\}
\label{43}\en
Repeating what we have already discussed for $\E$ we conclude that, if $\rho_f(x)$ and $\rho_g(x)$ are $C^\infty$ functions, then $\D(\mathbb{R})\subseteq\V$, which is therefore dense in $\Lc^2(\mathbb{R})$. Moreover, it is easy to check that $\V$ is closed under linear combinations: if $v_1(x), v_2(x)\in\V$, then $\alpha_1v_1(x)+\alpha_2v_2(x)\in\V$ as well, for all complex $\alpha_1,\alpha_2$. In our conditions we can check that $(\F_f,\F_g)$ are $\V$-quasi bases,
\be
\langle v,w\rangle=\sum_{n\geq0}\langle v,f_n\rangle\langle g_n,w\rangle=\sum_{n\geq0}\langle v,g_n\rangle\langle f_n,w\rangle,
\label{44}\en
for all $v(x), w(x)\in\V$, so that they are both complete in $\V$. 

Let us now observe that, even if, say, $f_n(x)\notin\Lc^2(\mathbb{R})$, for instance because its asymptotic behaviour is not the proper one, $af_n(x)$ can make sense, $a$ being some specific operator, even if $af_n(x)$ could not be in $\Lc^2(\mathbb{R})$. Or, maybe, $b^\dagger g(x)$ can be well defined even if $g_n(x)\notin\Lc^2(\mathbb{R})$ for some operator $b^\dagger$, even if  $b^\dagger g_n(x)\notin\Lc^2(\mathbb{R})$. This is exactly what happens in Section \ref{sectIII}. For this reason, it is natural to suppose the following: there exist two linear operators $a$ and $b$, together with a strictly increasing sequence $\{\alpha_n\}$, with $\alpha_0=0$, such that $af_n=\alpha_n f_{n-1}$, $b^\dagger g_n=\alpha_n g_{n-1}$, $n\geq1$, and $af_0=b^\dagger g_0=0$. Calling, as before, $\overline{\alpha}=\sup_n\alpha_n$ we know that $N(|z|)$ defined as in (\ref{23}) exists inside $C_{\overline{\alpha}}(0)$, a circle in the complex plane centered in the origin and of radius $\overline{\alpha}$. So far, there are not many differences with what stated by Theorem \ref{theo1}. But now, it is clear that it makes no sense to use conditions (\ref{22}), since these would refer to the norm of $f_n(x)$ and $g_n(x)$, which can be infinite here. Hence, no bound as those in (\ref{22}) is expected to hold now, and the convergence of the analogous of the $\varphi(z)$ and $\Psi(z)$ in (\ref{24}) is not granted, at all. For this reason we need to adopt a different strategy which is quite close to what one does when introducing distributions as continuous functionals on some special set of functions, see \cite{gel} for instance.

Let $v\in\V$, and let us consider the following series
\be
S_{f,v}(z)=\sum_{n\geq0}\,\frac{z^n}{\alpha_n!}\langle f_n,v\rangle, \qquad S_{g,v}(z)=\sum_{n\geq0}\,\frac{z^n}{\alpha_n!}\langle g_n,v\rangle.
\label{45}\en
Both these series converge, for all $v\in\V$, inside $C_{\overline{\alpha}}(0)$. In fact, 
$$
|\langle f_n,v\rangle|=|\langle c_n\rho_f,v\rangle|=|\langle c_n,\overline{\rho_f}v\rangle|\leq \|c_n\| \|\overline{\rho_f}v\|= \|\overline{\rho_f}v\|<\infty,
$$
since $\overline{\rho_f(x)}v(x)$ belongs to $\Lc^2(\mathbb{R})$ by the definition of $\V$. Hence
$$
|S_{f,v}(z)|\leq \|\overline{\rho_f}v\|\sum_{n\geq0}\,\frac{|z|^n}{\alpha_n!},
$$
which converges inside $C_{\overline{\alpha}}(0)$, as stated, independently of $v\in\V$. Similarly we find
$$
|S_{g,v}(z)|\leq \|\overline{\rho_g}v\|\sum_{n\geq0}\,\frac{|z|^n}{\alpha_n!},
$$
which is also converging inside $C_{\overline{\alpha}}(0)$ for all $v\in\V$. Then we introduce two vectors, $f(z)$ and $g(z)$, $z\in C_{\overline{\alpha}}(0)$, using the following definitions:
\be
\langle f(z),v\rangle=N(|z|)\,S_{f,v}(\overline{z})=N(|z|)\sum_{n\geq0}\,\frac{\overline{z}^n}{\alpha_n!}\langle f_n,v\rangle, \label{46}\en
and \be \langle g(z),v\rangle=N(|z|)\,S_{g,v}(\overline{z})=N(|z|)\sum_{n\geq0}\,\frac{\overline{z}^n}{\alpha_n!}\langle g_n,v\rangle. \label{47}\en
Of course, these equations do not give us the explicit analytic expressions for $f(z)$ and $g(z)$, but they can be used to define $f(z)$ and $g(z)$ {\em in a weak sense}. For that, we start introducing two maps, from $\V$ to $\mathbb{C}$, as follows:
\be
F(z)[v]=\langle f(z),v\rangle, \qquad G(z)[v]=\langle f(z),v\rangle.
\label{48}\en
It is clear that both $F(z)$ and $G(z)$ are linear functionals on $\V$. Next we introduce on $\V$ a topology $\tau_\V$ as follows: we say that a sequence $\{v_n(x)\}$ in $\V$ is $\tau_\V$-convergent to a certain $v(x)\in\Lc^2(\mathbb{R})$ if  $\{v_n(x)\}$ converges to $v(x)$ in the norm $\|.\|$, and if $\{\rho_j(x)\,v_n(x)\}$, $j=f,g$, are Cauchy sequences in $\|.\|$ and converge to $\rho_j(x)v(x)$. It is clear that, when this is true, $v(x)\in\V$. Hence, $\V$ is closed in $\tau_\V$.

Now, if we call $\V'$ the set of all the continuous functionals on $\V$, we can prove the following result:

\begin{prop}
	$F(z)$ and $G(z)$ both belong to $\V'$.
\end{prop}
\begin{proof}
We only need to prove that $F(z)$ and $G(z)$ are continuous. We will give the proof for $F(z)$, since that for $G(z)$ is completely analogous.

Let $\{v_k(x)\}$ be a sequence in $\V$  $\tau_\V$-convergent to a certain $v(x)\in\V$. We need to prove that $F(z)[v_k-v]\rightarrow0$, when $k$ diverges. Indeed, since for all $k,n\geq0$, $\langle f_n,v_k-v\rangle=\langle c_n,\overline{\rho_f}(v_k-v)\rangle$, we have
$$
|F(z)[v_k-v]|=\left|N(|z|)\sum_{n\geq0}\,\frac{\overline{z}^n}{\alpha_n!}\langle f_n,v_k-v\rangle\right|\leq N(|z|)\sum_{n\geq0}\,\frac{|z|^n}{\alpha_n!}|\langle f_n,v_k-v\rangle|,
$$
so that
$$
|F(z)[v_k-v]|\leq N(|z|)\sum_{n\geq0}\,\frac{|z|^n}{\alpha_n!}|\langle c_n,\overline{\rho_f}(v_k-v)\rangle|\leq \|\overline{\rho_f}(v_k-v)\|N(|z|)\sum_{n\geq0}\,\frac{|z|^n}{\alpha_n!}\rightarrow0,
$$
for $k\rightarrow\infty$, due to the definition of $\tau_\V$ and to the convergence of the series inside $C_{\overline{\alpha}}(0)$.

\end{proof}

The continuous functionals $F(z)$ and $G(z)$ are what we call {\em weak bi-coherent states} (wbcs). The reason is contained in the following proposition, where the analogous of (\ref{25}) and (\ref{27}), two of the essential properties of any (bi-)coherent states, are proven. We need first to introduce an useful subspace of $\V$:
\be\V_0=\left\{w(x)\in\V:\, w(x)\in D(a^\dagger)\cap D(b), \quad a^\dagger w(x), bw(x)\in \V\right\}
\label{49}\en
Of course the explicit form of $\V_0$ depends also on $a$ and $b$. We will show later that, at least in the situation considered in Section \ref{sectIII}, $\V_0$ is dense in $\Lc^2(\mathbb{R})$ since it contains $\D(\mathbb{R})$.

\begin{prop}\label{prop2}
	
	The wbcs satisfy the following properties:
	
	(i) for all $v(x)\in \V_0$ we have
	\be
	\langle v,af(z)\rangle=z\langle v,f(z)\rangle, \qquad \langle v,b^\dagger g(z)\rangle=z\langle v,g(z)\rangle,
	\label{410}\en
	for all $z\in C_{\overline{\alpha}}(0)$.
	
	(ii) Suppose that a measure $d\lambda(r)$ does exist such that
	\be
	\int_0^{\overline{\alpha}} d\lambda(r) r^{2k}=\frac{(\alpha_k!)^2}{2\pi},
	\label{411}\en
	for all $k\geq0$. Then, putting $z=re^{i\theta}$ and calling $d\nu(z,\overline z)=N(r)^{-2}d\lambda(r)d\theta$, we have
	\be
	\int_{C_{\overline{\alpha}}(0)}\left<v,f(z)\right>\left<g(z),w\right>d\nu(z,\overline z)=
	\int_{C_{\overline{\alpha}}(0)}\left<v,g(z)\right>\left<f(z),w\right>d\nu(z,\overline z)=
	\left<v,w\right>,
	\label{412}\en
	for all $v,w\in\V$.

\end{prop}

\begin{proof}
	We start proving (i). Since, by assumption, $v(x)\in\V_0$, then   $v(x)\in D(a^\dagger)$ and $a^\dagger v(x)\in \V$. Hence we have
	$$
	\langle v,af(z)\rangle=\langle a^\dagger v,f(z)\rangle=N(|z|)\sum_{n\geq0}\,\frac{z^n}{\alpha_n!}\langle a^\dagger v, f_n\rangle=N(|z|)\sum_{n\geq0}\,\frac{z^n}{\alpha_n!}\langle v,af_n\rangle.
	$$
	using (\ref{46}). Now, since $af_0=0$ and $af_n=\alpha_n f_{n-1}$, $n\geq1$, the first equality in (\ref{410}) easily follows. The second equality can be proved in a similar way.
	
	The proof of (ii) does not differ significantly from that of Theorem \ref{theo1}, and will not be given here. We only stress that the crucial aspect of the proof is the fact that $(\F_f,\F_g)$ are $\V$-quasi bases.

\end{proof}

This proposition shows that our wbcs produce, together, a resolution of the identity on $\V$. It also shows that they are {\em weak} eigenstates of the lowering operators $a$ and $b^\dagger$, but this is only granted on $\V_0$, and not on $\V$.

In the last part of this section we apply the general results deduced so far to what we have seen in Section \ref{sectIII}, to check if and when wbcs can be defined. Formula (\ref{312}) shows that we are indeed in the situation described in (\ref{41}), with $\varphi_n(x)$, $\rho_A(x)$, $\Psi_n(x)$ and $\rho_B(x)$ identified respectively with $f_n(x)$, $\rho_f(x)$, $g_n(x)$ and $\rho_g(x)$. The o.n. basis $\F_c$ in (\ref{41}) is constructed as in (\ref{313}). Now, since $A$ and $B$ are pseudo-bosonic operators, $\alpha_n$ is easily identified: $\alpha_n=\sqrt{n}$, which implies that $\overline{\alpha}=\infty$. Condition (\ref{411}) on the measure $d\lambda(r)$ reads therefore
$$
	\int_0^{\infty} d\lambda(r) r^{2k}=\frac{k!}{2\pi},
$$
as in Section \ref{bcsexample}. Hence $d\lambda(r)=\frac{1}{\pi}e^{-r^2}r\,dr$, which is the same result we have found in (\ref{322}). This means that the forms $F(z)[v]$ and $G(z)[v]$ are well defined for all $v\in\V$ and for all $z\in\mathbb{C}$, and that (\ref{412}) is satisfied with these particular choices.

As for (\ref{410}), it is useful to show that $\V_0$ is {\em large enough}. Otherwise these equalities are not particularly useful. For instance, if $\V_0$ would contain only the zero vector, (\ref{410}) is surely true, but it becomes completely trivial. Luckily enough, this is not the case here. Indeed, we can check that $\D(\mathbb{R})\subseteq\V_0$. In fact, we start observing that $\V$ coincides, in this case, with the set $\E$ in (\ref{316}) which, as already noticed, contains $\D(\mathbb{R})$. Therefore $\D(\mathbb{R})\subseteq\V$. Now, let $v_0(x)\in\D(\mathbb{R})$. Then $v_0(x)\in\V$. Moreover, $v_0(x)$ belongs to the domain of both $A^\dagger$ and $B$. This is because, for instance
$$
Bv_0(x)=\left(-\frac{d}{dx}+w_B(x)\right)v_0(x)=-v_0'(x)+w_B(x)v_0(x),
$$
which is again in $\D(\mathbb{R})$ being the difference of two functions, $v_0'(x)$ and $+w_B(x)v_0(x)$, which are both in $\D(\mathbb{R})$. In particular, we are using here the fact that $w_B(x)$ is $C^\infty$. So the conclusion is clear: $Bv_0(x)$ and $A^\dagger v_0(x)$ are both in $\D(\mathbb{R})$, so that they are also both in $\V$. Hence $v_0(x)$ belongs to $\V_0$, which for this reason is dense in $\Lc^2(\mathbb{R})$. 

We can conclude that, for the operators considered in Section \ref{sectIII}, the statements proved in Proposition \ref{prop2} hold in {\em large sets}, so that wbcs can be defined under very general conditions on $w_A(x)$ and $w_B(x)$.

\section{Conclusions}\label{sectconl}

This paper is a continuation of our previous analysis on deformed canonical commutation relations and its related pseudo-bosons and bi-coherent states. In particular, we have shown how to extend the original approach outside $\Lc^2(\mathbb{R})$, keeping untouched the possibility of defining eigenstates of certain number-like operators, manifestly non self-adjoint, but opening to the possibility that these states are not square-integrable. The example discussed in this paper is constructed in terms of two different, but related, superpotentials. For this reason the approach described here could be relevant for supersymmetric quantum mechanics. But this is possible only when considering choices of $w_A(x)$ and $W_B(x)$ not giving rise to pseudo-bosonic operators. Otherwise, as discussed before, the super-symmetric partner of the starting Hamiltonian coincides with this, except that for an additive constant.

In the second part of the paper we have shown how bi-coherent states can be introduced for the PBSs considered here, and we have constructed a general framework to consider a situation in which $\Lc^2(\mathbb{R})$ is not sufficient, in the same line as what was done recently in \cite{bag2020JPA}. In particular we have given conditions on the superpotentials for the wbcs to be well defined, and in order for them to satisfy analogous properties as those of ordinary coherent and bi-coherent states.

\section*{Acknowledgements}

The author acknowledges partial support from Palermo University and from G.N.F.M. of the INdAM.

\end{document}